\documentclass[a4paper]{article}

\usepackage{amsthm, amsfonts, amsmath}
\usepackage{graphicx}
\usepackage{appendix}
\usepackage{algorithm}
\usepackage{algorithmicx}
\usepackage[noend]{algpseudocode}
\usepackage{booktabs}
\usepackage{enumerate}
\usepackage{mathtools}
\usepackage{xspace} 
 \usepackage[numbers]{natbib}
\usepackage[hidelinks]{hyperref}
\usepackage[colorinlistoftodos,prependcaption,textsize=tiny, disable]{todonotes}

\RequirePackage[T1]{fontenc} \RequirePackage[tt=false, type1=true]{libertine} \RequirePackage[varqu]{zi4} \RequirePackage[libertine]{newtxmath}

\newcommand{\phil}[1]{\todo[backgroundcolor=green!25]{#1}}

\newtheorem{lemma}{Lemma}
\newtheorem{theorem}[lemma]{Theorem}

\newtheorem{definition}[lemma]{Definition}

\newtheorem{remark}[lemma]{Remark}

\newcommand{\bigO}{\smash{\ensuremath{O}}}
\newcommand{\tilO}{\smash{\ensuremath{\widetilde{O}}}}
\newcommand{\tilOm}{\smash{\ensuremath{\widetilde{\Omega}}}}
\newcommand{\tilT}{\smash{\ensuremath{\widetilde{\Theta}}}}

\newcommand{\hybrid}{\ensuremath{\mathsf{HYBRID}}\xspace}
\newcommand{\HYBRID}{\ensuremath{\mathsf{HYBRID}}\xspace}

\newcommand{\hybridpar}[2]{\ensuremath{\mathsf{HYBRID}(#1,#2)}}
\newcommand{\hybridparbig}[2]{\ensuremath{\mathsf{HYBRID}\big(#1,#2\big)}}
\newcommand{\LOCAL}{\ensuremath{\mathsf{LOCAL}}\xspace}
\newcommand{\CONGEST}{\ensuremath{\mathsf{CONGEST}}\xspace}

\newcommand{\NCC}{\ensuremath{\mathsf{NCC}}\xspace}

\newcommand{\CC}{\ensuremath{\mathsf{CLIQUE}}\xspace}

\newcommand{\E}{\mathbb{E}}
\newcommand{\eps}{\varepsilon}
\newcommand{\calA}{\mathcal{A}}
\newcommand{\calB}{\mathcal{B}}

\newcommand{\calG}{\mathcal{G}}
\newcommand{\calH}{\mathcal{H}}

\newcommand{\calS}{\mathcal{S}}
\newcommand{\calT}{\mathcal{T}}

\newcommand*{\medcup}{\mathbin{\scalebox{1.5}{\ensuremath{\cup}}}}

\newcommand{\NQ}{\mathcal{NQ}}
\newcommand{\NQpar}{\ensuremath{\mathcal{NQ}(G,k,\gamma)}\xspace}

\newcommand{\p}{\ensuremath{\!+\!}}
\newcommand{\m}{\ensuremath{\!-\!}}

\DeclareMathOperator{\polylog}{polylog}

\DeclareMathOperator{\hop}{hop}
\DeclareMathOperator*{\argmin}{arg\,min}
\DeclareMathOperator*{\argmax}{arg\,max}

\begin{document}

\title{Towards Universally Optimal Shortest Paths Algorithms in the Hybrid Model}
\date{}
\author{Philipp Schneider, University of Bern, Switzerland\\ (philipp.schneider2@unibe.ch)}
\maketitle

\begin{abstract}
	
A drawback of the classic approach for complexity analysis of distributed graph problems is that it mostly informs about the complexity of notorious classes of ``worst case'' graphs. Algorithms that are used to prove a tight (existential) bound are essentially optimized to perform well on such worst case graphs. However, such graphs are often either unlikely or actively avoided in practice, where benign graph instances usually admit much faster solutions. 
	
To circumnavigate these drawbacks, the concept of \textit{universal} complexity analysis in the distributed setting was suggested by [Kutten and Peleg, PODC'95] and actively pursued by [Haeupler et al., STOC'21]. Here, the aim is to gauge the complexity of a distributed graph problem \textit{depending} on the given graph instance. The challenge is to identify and understand the graph property that allows to accurately quantify the complexity of a distributed problem on a given graph.  
\phil{put this sentence into the intro in the submission (commented out)} 

In the present work, we consider distributed shortest paths problems in the \hybrid model of distributed computing, where nodes have simultaneous access to two different modes of communication: one is restricted by \textit{locality} and the other is restricted by \textit{congestion}. We identify the graph parameter of \textit{neighborhood quality} and show that it accurately describes a universal bound for the complexity of certain class of shortest paths problems in the \hybrid model.
\end{abstract}

\section{Introduction}

\textit{Remark: This is a preprint article that focuses on technical contributions in order to facilitate fast scientific exchange. In particular, this version of the article comes without a dedicated account of related work or an extensive introduction, which will be added shortly in an updated version.}

\subsection{The \hybrid model}

For the formal definition of the \hybrid model, we rely on  the concept of \textit{synchronous message passing} \cite{Lynch1996}, where nodes exchange messages and conduct local computations in synchronous rounds. Note that synchronous message passing focuses on the \textit{round complexity}, i.e., the number of rounds required to solve a distributed problem, and therefore nodes are considered computationally unbounded.

\begin{definition}[Synchronous Message Passing, cf.\ \cite{Lynch1996}] 
	\label{def:sync_msg_passing}
	Let $V$ be a set of $n$ nodes with unique identifiers ID\smash{$:V \!\to\! [n] \!\stackrel{\text{def}}{=}\! \{1, \ldots , n\}$}. Time is slotted into discrete rounds. Nodes wake up synchronously and start executing an algorithm $\calA$, which determines each nodes behavior in each round consisting of the following steps. First, all nodes receive the set of messages addressed to them in the last round. Second, nodes conduct computations based on their current state and the set of received messages to compute their new state (randomized algorithms also include the result of a random function). Third, based on the new state, the next messages are sent.
\end{definition}

The aim of the  \hybrid model is to reflect the fundamental concepts of \textit{locality} and \textit{congestion} to capture the nature of distributed systems that combine {both} {physical} and {logical} networks.
%
%

\begin{definition}[cf. \cite{Augustine2019}]
	\label{def:hybrid}
	The \hybridpar{\lambda}{\gamma} model is a synchronous message passing model (Def.\ \ref{def:sync_msg_passing}), subject to the following restrictions. \emph{Local mode:} nodes may send a message per round of maximum size $\lambda$ bits to each of their neighbors in a connected graph. \emph{Global mode:} nodes can send and receive messages of total size at most $\gamma$ bits per round to/from any other node(s) in the network. If these restrictions are violated a strong adversary\footnote{The strong adversary knows the states of all nodes, their source codes and even the outcome of all random functions.} selects the messages that are delivered.
\end{definition}

The parameter $\lambda$ restricts the bandwidth over {edges} in the local network, and $\gamma$ restricts the amount of global communication of {nodes}. Notably, the classical models of distributed computing are covered by this model as marginal cases: \LOCAL and \CONGEST are given by $\gamma = 0$ and $\lambda = \infty$ and $\lambda \in \bigO(\log n)$, respectively. The \CC and \NCC models are given by $\lambda = 0$ and $\gamma \in \bigO(n \log n)$ (due Lenzens  routing algorithm \cite{Lenzen2013})  and $\gamma \in \bigO(\log^2 n)$, respectively.

Of particular practical and theoretical interest are non-marginal parameterizations of \hybridpar{\lambda}{\gamma} that push both communication modes to one extreme end of the spectrum. More specifically, to model the high bandwidth of physical connections we leave the size of local messages unrestricted. To model the severely restricted global bandwidth of shared logical networks, we allow only $\polylog n$ bits of global communication per node per round. Formally, we define the ``standard'' hybrid model as combination of the standard \LOCAL \cite{Peleg2000} and node capacitated clique \cite{Augustine2019} models: $\hybrid := \hybridparbig{\infty}{\tilO(1)}$.

\subsection{Preliminaries}
\label{sec:preliminaries}

We continue with some definitions, conventions and nomenclature that we will use in the following.

\begin{definition}[$(k,\ell)$-{Shortest Paths} ($(k,\ell)$-SP) Problem]
	\label{def:kSSP}
	We are a given subsets of $V$ of $k$ source and $\ell$ target nodes (not necessarily disjoint) in a graph $G=(V,E)$. Every target $t$ has to learn $d_G(s,t)$ for all sources $s$. In the $\alpha$-approximate version of the problem for \emph{stretch} $\alpha \geq 1$, every target node $t$ has to learn values $\tilde{d}(s,t)$ such that $d(s,t)\leq \tilde{d}(s,t)\leq \alpha d(s,t)$ for all source nodes $s$.
\end{definition}

Given that $\ell = n$ we talk about the $k$-sources shortest paths problem ($k$-SSP). Further special cases are the \textit{all-pairs shortest paths problem} (APSP) for $k=\ell=n$ and the single-source shortest paths problem (SSSP) $\ell = n, k=1$. Note that the local communication graph and the input graph for the graph problem are the same, which is a standard assumption for distributed models with graph-based communication (like \LOCAL and \CONGEST).

%

Our algorithms are randomized, i.e., they are supposed successfully compute the solution of a problem with probability $p>0$ for any problem instance. We aim for success \textit{with high probability} (w.h.p.), which means $p \geq 1-\frac{1}{n^c}$ for an arbitrary constant $c>0$. We write i.i.d.\ if we pick elements from some set \textit{independently, identically distributed}. 

In this work, logarithm functions without subscript are generally to the base of two, i.e., \smash{$\log \stackrel{\text{def}}{=} \log_2$}. Sometimes we write $\polylog n$ to describe terms of the form $q(\log n)$ where $q$ is a polynomial. We abbreviate sets of the form $\{1, \dots, k\},  k \in \mathbb{N}$ with $[k]$. We will often neglect logarithmic factors in $n$ using the soft $\tilO$-notation.

We consider undirected, connected communication graphs $G = (V,E)$. Edges have {weights} $w: E \to [W]$, where $W$ is at most polynomial in $n$, thus the weight of an edge and of a simple path fits into a $\bigO(\log n)$ bit message. A graph is considered {unweighted} if $W=1$. Let $w(P) = \sum_{e \in P}w(e)$ denote the length of a path $P \subseteq E$. 

Then the \emph{distance} between two nodes $u,v \in V$ is
{$
	d_G(u,v) := \!\min_{\text{$u$-$v$-path } P} w(P).
	$}
A path with smallest length between two nodes is called a \emph{shortest path}.
Let $|P|$ be the number of edges (or \emph{hops}) of a path $P$.
We define the \emph{$h$-hop limited distance} from $u$ to $v$ as
$
d_{G,h}(u,v) := \!\!\min_{{\text{$u$-$v$-path } P, |P| \leq h }}\, w(P).
$
If there is no $u$-$v$ path $P$ with $|P|\leq h$ we define $d_{G,h}(u,v) := \infty$.

The \emph{hop-distance} between two nodes $u$ and $v$ is defined as
{$
	\hop_G(u,v) := \!\min_{\text{$u$-$v$-path } P} |P|.
	$ }
We generalize this for sets $U,W \subseteq V$
{$
	\hop_G(U,W) := \!\min_{u \in U, w \in W} \hop_G(u,w)
	$}  (whereas $\hop_G(v,v) := 0$).
The \emph{diameter} of $G$ is defined as
{$
	D_G:= \max_{u,v \in V} \hop_G(u,v).
	$}
For $v\in V, r\in \mathbb N$ we define the $h$-\textit{hop neighborhood} (or $h$-hop \textit{ball}) of $v$:
$\calB_G(v,h) := \{u \in V \mid \hop(u,v) \leq h\}.$
We generalize this for node sets $V' \subseteq V$, as well:
$\calB_G(V',h) := \medcup_{v \in V'} \calB_G(v,h).$
We drop the subscript $G$, whenever $G$ is clear from the context.

\subsection{Universal Optimality}

Our concept of universal optimality adheres closely to the one by \cite{Haeupler2021} which bases itself on a description by \cite{Kutten1995}. Consider a graph problem $\Pi$. Then a problem instance of $(G,I)\in \Pi$ consists of a graph $G$ together with a (distributed) problem input $I$.
In the case of the $(k, \ell)$-SP problem, $I$ assigns each node an ID in $[n]$, the IDs of its neighbors as well as the weights of the corresponding edges. The set of source and target nodes is also considered part of $I$, i.e., each node knows whether it is a source or target (or both) but has initially no such information about others.

For the purpose of comparison (see also \cite{Haeupler2021}), we start by defining an even stronger concept, where an algorithm would be called \textit{instance optimal} if it would be competitive with any algorithm optimized for a certain graph $G$ \textit{and} instance $I$.

\begin{definition}[Instance Optimality, see e.g. \cite{Kutten1995, Haeupler2021}]
	\label{def:inst_opt}
	Let $\calA$ be an algorithm that correctly computes the solution to some distributed graph problem $\Pi$ with probability at least $p>0$ in some computational model $\mathcal M$ and takes $T_{\calA}(G,I)$ rounds for $(G,I) \in \Pi$. Then $\calA$ is called an \emph{instance optimal} model $\mathcal M$ algorithm with competitiveness $C$ (omitted for $C \in \tilO(1)$) if the following holds. For all $(G,I) \in \Pi$ and for all algorithms $\calA'$ that solve $\Pi$ with probability at least $p$, we have
	$$T_{\calA}(G,I) \leq C \cdot  T_{\calA'}(G,I).$$
\end{definition}

Unfortunately, the concept of instance optimality is not very interesting in the \hybrid model unless the problem $\Pi$ admits very fast solutions in general. This is due to fact that in the \hybrid model nodes can detect extremely fast if they live in a specific problem instance $(G,I) \in \Pi$ by using the global network and exploit this knowledge accordingly (compare this to the analogous argument  for \CONGEST by \cite{Haeupler2021}). For an illustration, consider the following algorithm $\calA'$ that is instance-optimized for $(G,I)$.

The instance $(G,I)$ and a solution for that instance (i.e., the output that each node with a given ID makes) is all hard coded into $\calA'$. Then each node checks locally if its ID equals an ID in $(G,I)$ and checks if all the input data it obtained equals the local data assigned to that ID by $I$. If that is the case, it outputs a 1 else a 0. The $n$-wise AND of all outputs of all nodes can be made public knowledge deterministically in just $\bigO(\log n)$ rounds \cite{Augustine2019}. If the result is 1, then all nodes output the hard coded result. Else all nodes together run the trivial $O(D_G)$ round \LOCAL algorithm to solve the problem instance.

So unless a graph problem $\Pi$ has $\tilO(1)$ complexity in \hybrid in general, instance optimality with competitiveness $\tilO(1)$ is unattainable for an algorithm that is \textit{oblivious} to $(G,I)$. For $k$-SSP there is in fact a polynomial (existential) lower bound of \smash{$\tilOm\big(\!\sqrt k\big)$} in \hybrid due to \cite{Kuhn2022}, so unless $k \in \tilO(1)$ there can be no instance optimal solution with competitiveness $\tilO(1)$. 
A more fruitful concept is to try to design algorithms that are competitive with the best algorithms that ``know'' the graph $G$ but not $I \in \Pi_G := \{I \mid (G,I) \in \Pi \}$. 
Such an algorithm is called \textit{universally optimal}, in the sense that it can ``adapt'' to the topology of $G$. 
Formally, we define this as follows.

\begin{definition}[Universal Optimality, see \cite{Kutten1995, Haeupler2021}]
	\label{def:uni_opt}
	Let $\calA$ be an algorithm that correctly computes the solution to some distributed graph problem $\Pi$ with probability at least $p>0$ in some computational model $\mathcal M$ and takes $T_{\calA}(G,I)$ rounds for $(G,I) \in \Pi$. Then $\calA$ is called a \emph{universally optimal} model $\mathcal M$ algorithm with competitiveness $C$ (omitted for $C \in \tilO(1)$) if the following holds. For all graphs $G$ and for all algorithms $\calA'$ that solve $\Pi$ with probability at least $p$, we have $$\max\limits_{I \in \Pi_G}T_{\calA}(G,I) \leq C \cdot \max\limits_{I \in \Pi_G} T_{\calA'}(G,I).$$
\end{definition}

\subsection{Contributions}

In this work, we give algorithms for shortest path problems with $k$ source nodes and $\ell$ target nodes ($(k, \ell)$-SP, see Definition \ref{def:kSSP}) in the \hybridpar{\infty}{\gamma} model that are universally optimal for certain ranges of $k$ and $\ell$.
For this purpose, in Section \ref{sec:neigh_qual}, we introduce a graph parameter $\NQpar$ called \textit{neighborhood quality} that, roughly speaking, describes the minimum number of nodes that any node has within a certain neighborhood. The parameter $\NQpar$ depends on $G$, the number of sources $k$ and the global capacity $\gamma$, takes values of at most $\bigO\big( \!\sqrt n\big)$ (see Lemma \ref{lem:sqrt_k_upper_bound}). 

In Section \ref{sec:upper_bound} we show that there exists an algorithm that solves the $(k,\ell)$-SP problem in $\tilO\big(\NQpar\big)$ rounds with stretch $1\p \eps$ for $\ell \in \tilO(1)$ and stretch $3 \p \eps$ for $\ell \leq \NQpar$ i.i.d.\ random target nodes (where $\eps >0$ is an arbitrary constant), see Theorem \ref{thm:uni_upper_bound}.
In Section \ref{sec:lower_bound} we match the upper bound for deterministic target nodes by giving a lower bound which shows that any polynomial approximation of the $(k,1)$-SP problem takes $\tilO\big(\NQpar\big)$ rounds even if all nodes are aware of $G$, see Theorem \ref{thm:universal_lower_bound}. From Theorems \ref{thm:uni_upper_bound}  and \ref{thm:universal_lower_bound} we can deduce the following 

\begin{theorem}
	There exists a \textit{universally optimal} \hybridpar{\infty}{\gamma} model algorithm that solves the $(k,\ell)$-SP problem in $\tilO\big(\NQpar\big)$ with stretch $1 \p \eps$ for $\ell \in \tilO(1)$ w.h.p.
\end{theorem}

\begin{proof}
	Let $\Pi$ be the problem of  solving $(k,1)$-SP with polynomial stretch.
	By Theorem \ref{thm:uni_upper_bound}, there exists an algorithm $\calA$ that solves $\Pi$ w.h.p.\ in $T_{\calA}(G,I) \in \tilO\big(\NQpar\big)$ rounds for any $(G,I) \in \Pi$.	
	Conversely, we show that  any algorithm $\calA'$  takes at least $\tilOm\big(\NQpar\big)$ rounds to solve $\Pi$ on a given graph $G$ with constant probability. 
	
	Assume that all nodes know $G$ and let $\hat \calA$ be the algorithm that locally selects\footnote{Each node enumerates all $I \in \Pi_G$ and algorithms $\calA'$ and determines $\calA'$ that minimizes $\max_{I \in \Pi_G}T_{\calA'}(G,I)$ using unlimited local computation (cf.\ Def.\ \ref{def:hybrid}). This argument is only for the non-constructive lower bound and no such computations are necessary for the algorithmic upper bound.} the algorithm $\calA'$ that minimizes $\max_{I \in \Pi_G}T_{\calA'}(G,I)$ and then executes $\calA'$. By design, $\hat \calA$  takes $\min_{\calA'} \max_{I \in \Pi_G}T_{\calA'}(G,I)$ rounds. However, by Theorem \ref{thm:universal_lower_bound}, we also know that $\hat \calA$ must still take at least $\tilOm\big(\NQpar\big)$ rounds. Putting it all together we obtain
	\[
		\tilOm\big(\NQpar\big) \ni T_{\hat \calA}(G,I)  = \min_{\calA'} \max_{I \in \Pi_G}T_{\calA'}(G,I) \leq 	T_{\calA}(G,I) \in \tilO\big(\NQpar\big).
	\]
	Therefore, the round complexity of any algorithm $\calA'$ (in particular those optimized for $G$) and the of algorithm $\calA$ are sandwiched between terms $\tilT\big(\NQpar\big)$, i.e., they differ only by a factor $C \in \tilO(1)$, thus $\calA$ is universally optimal (Definition \ref{def:uni_opt}). Note that the upper and lower bound of $\tilT\big(\NQpar\big)$ holds for any $\ell \in \tilO(1)$.
\end{proof}

We conclude that the $(k,1)$-SP problem has a universal bound  $\tilT\big(\NQpar\big)$ on any graph instance, which we consider as a first step towards universally optimal shortest paths algorithms for more general cases. We point out that our results go somewhat beyond this, in the sense that clearly the lower bound of $\tilOm\big(\NQpar\big)$ holds for the general $k$-SSP problem as well. At the same time the upper bound $\tilOm\big(\NQpar\big)$ holds for a larger number of targets $\ell \leq \NQpar$ in case they are selected randomly. The interesting open question is whether there is a matching lower bound for random sources or a matching upper bound for a larger number of fixed sources, which would show the universality of the bound $\tilT\big(\NQpar\big)$ for a broader spectrum of the $(k,\ell)$-SP problem.\phil{mention routing as technical contribution}

\section{The Graph Parameter Neighborhood Quality}
\label{sec:neigh_qual}

We start by giving a fundamental graph parameter that describes the complexity of algorithms solving the $(k,\ell)$ shortest paths problem  ($(k,\ell)$-SP) where all nodes obtain the graph $G$ as part of their input. Leaning on the nomenclature used by previous work (see, e.g., \cite{Haeupler2021}), we call this the neighborhood quality $\NQpar$. Intuitively, $\NQpar$ describes how large the neighborhood of each node within a certain distance is that such a node can utilize on in order to communicate globally with others. This distance depends on $k$ and $\gamma$, where $k$ roughly reflects the `ìnformation'' (measured in Shannon entropy \cite{Shannon1948}) each node has to learn and $\gamma$ restricts the information a single node can receive over large distances per round.

To reflect the runtime of shortest paths algorithms in $G$, the parameter $\NQpar$ is defined inversely. Intuitively, neighborhoods of ``higher quality'' imply that  $\NQpar$ is smaller. Since the trivial solution of $D_G$ rounds is always possible using the local network even if nodes are required to learn huge amounts of information, $\NQpar$ has to be, in effect, upper bounded by $D_G$, which roughly means that neighborhoods play a bigger role on graphs with large diameter (on which global problems become interesting). 

\begin{definition}
	\label{def:neigh_qual}
	Let $G=(V,E)$ be a local graph and $k \in [n]$. For $v \in V$ let $N(d,v) = |\calB(v,d)|$ be the size of the $d$-hop neighborhood  of $v$. Then we define the smallest neighborhood in distance $d \in [D_G]$ as $N(d) := \min_{v \in V} N(d,v)$.
	We define the neighborhood quality in the \hybridpar{\infty}{\gamma} with respect to some $k \in [n]$ as
	\[
		\NQ(G,k,\gamma) := \min\limits_{d \in [D_G]} \max \Big(\tfrac{k}{N(d)\gamma}, d\Big).
	\]
	An equivalent definition that is slightly longer but sometimes useful for explanation is
	\[
	\NQ(G,k,\gamma) := \max_{v \in V}\min\limits_{d \in [D_G]} \max \Big(\tfrac{k}{N(d,v)\gamma}, d\Big).
	\]
\end{definition}

An important property of $\NQ(G,k,\gamma)$ is that it roughly strikes a balance between radius and size of the neighborhood of any node within that radius. This is reflected in the following technical lemma, which we will need later on.

\begin{lemma}
	\label{lem:neigh_qual_min_neigh_size}
	Let $G=(V,E)$ be a graph and $k \in[n]$. Let $d'$ be such that the outer minimum of $\NQ(G,k,\gamma)$ in Def.\ \ref{def:neigh_qual} is minimized and presume $d' < D_G$. Then we have
	\begin{enumerate}[(a)]
		\item $d' \p 1\geq k/N(d')\gamma$ \quad (equivalently: $N(d')\geq k/(d'\p 1)\gamma$)
		\item $d' \leq \NQ(G,k,\gamma) \leq d'\p 1 \leq D_G$ \quad (implies: $N(d')\geq k/(\NQpar\p 1)\gamma$)
	\end{enumerate}
\end{lemma}

\begin{proof}
	Note that as $d$ increases by 1, $N(d)$ increases by at least 1, because the neighborhood of each node within radius $d$ gets at least one node larger as $d < D_G$ and $G$ is connected. This implies that \smash{$\tfrac{k}{N(d)\gamma}$} strictly decreases in $d$. Let us assume (for a contradiction) that \smash{$d' \p 1< \tfrac{k}{N(d')\gamma}$}. This implies that 
	$$\max \Big(\tfrac{k}{N(d')\gamma}, d'\Big) = \max \Big(\tfrac{k}{N(d')\gamma}, d' \p 1\Big) > \max \Big(\tfrac{k}{N(d'+1)\gamma}, d' \p 1\Big),$$
	which is a contradiction since $d'$ minimizes the term on the left, thus \smash{$d' \p 1 \geq \tfrac{k}{N(d')\gamma}$}. The claim $d' \leq \max \big(\tfrac{k}{N(d')\gamma}, d'\big) \leq d'+1$ follows immediately.
\end{proof}

The first indication that $\NQ(G,k,\gamma)$ is a suitable parameter to describe a universal bound for $k$-SSP (which we prove in the subsequent sections) is obtained by relating it to the existential lower bound of \smash{$\tilOm\big(\!\sqrt k\big)$} in the standard \hybrid model (i.e., $\gamma \in \tilO(1)$) by \cite{Kuhn2020, Augustine2020a} that describes worst case instances for the $k$-SSP problem.

\begin{lemma}
	\label{lem:sqrt_k_upper_bound}
	Let $G$ be a graph and $k \in [n]$. We have $1 \leq \NQ(G,k,\gamma) \leq  \sqrt{k/\gamma} +1.$	
\end{lemma} 
	
\begin{proof}
	The value $\NQ(G,k,\gamma)$ is always at least 1 for the given parameter range of $d$. Looking more closely at Definition \ref{def:neigh_qual} we see that $\NQ(G,k,\gamma)$ is largest if there is a node that has only a very sparse neighborhood $N(d) \in \Theta(d)$. We consider the worst case where $G$ is a path, thus $N(d,v) =d$ for one endpoint $v$ of that path. Assume that $d$ and $N(d,v)$ can attain real values, then \smash{$\min_{d \leq D_G} \max \big(\tfrac{k}{d\gamma}, d\big) $} is minimized if 	$\tfrac{k}{d\gamma} = d$, which implies \smash{$d = \sqrt{k/\gamma}$}.
	For integral $d$ it is \smash{$\NQpar \leq  \max \big(\tfrac{k}{\lceil d\rceil \gamma}, \lceil d\rceil\big) \leq \sqrt{k/\gamma} +1$}.
\end{proof}

The lemma implies that the neighborhood quality is upper bounded by the existential lower bound \smash{$\NQ(G,k,\gamma) \in \tilO\big(\!\sqrt k\big)$} (where $\gamma \in \tilO(1)$ in the standard \hybrid model). The proof also shows that there is in fact a problem instance $G,k$ with a linear neighborhood such that \smash{$\NQ(G,k) \in \tilOm\big(\!\sqrt k\big)$}. In fact, graphs which feature an isolated, long path have frequently been used to obtain existential lower bounds for shortest paths problems in the \hybrid model \cite{Augustine2020a,Kuhn2020, Kuhn2022}.

In the following, we often consider the quality of the neighborhood of a single node $v \in V$, which we define as follows.

\begin{definition}
	\label{def:neigh_qual_node}
	Let $G=(V,E)$ be any fixed graph and consider any fixed $k \in[n]$ for which the $k$-SSP problem needs to be solved on $G$. For $v \in V$ we define  $$\NQ(v) = \min\limits_{d \in [D_G]} \max \Big(\tfrac{k}{N(v,d)\gamma}, d\Big).$$
	Notice that we use the size $N(d,v)$ (cf.\ Definition \ref{def:neigh_qual}) of the neighborhood within $d$ hops of $v$, thus $\NQ(v)$ describes the neighborhood quality of a node $v$. In fact, we have $\NQ(G,k,\gamma) = \max_{v \in V} \NQ(v)$, see Definition \ref{def:neigh_qual}.
	We often use $d_v$ for the value that optimizes the outer minimum in the definition of $\NQ(v)$. In Section \ref{sec:lower_bound} on universal lower bounds, we will often relate to $v := \argmax_{u \in V} \NQ(u)$.
\end{definition}	

Be aware that $\NQ(v)$ also depends on the problem instance $G,k$ and $\gamma$, but we omit to express this specifically for brevity. As a warm up, we show that the parameter $\NQ(G,k,\gamma)$ can be computed efficiently in the \hybrid model, in particular, we give an algorithm that takes roughly $\NQ(G,k,\gamma)$ rounds.

\begin{lemma}
	\label{lem:compute_neigh_qual}
	The parameter $\NQ(G,k,\gamma)$ can be computed and made known to all nodes in the network in $\tilO\big(\NQ(G,k,\gamma)\big)$ rounds in the \hybrid model.	
\end{lemma}

\begin{proof}
	The idea is that each node locally computes $\NQ(v)$ and then all nodes globally compute the maximum of those values to obtain $\NQ(G,k,\gamma)$ (see Definition \ref{def:neigh_qual}).  The first step works by each node exploring the local network to increasing depth $d$ and locally computing \smash{$\max \big(\tfrac{k}{N(v,d)\gamma}, d\big)$} which takes $d$ rounds using \LOCAL. 		
	To ensure that we do not look too deep, i.e., beyond \smash{$\bigO\big(\NQ(G,k,\gamma)\big)$}, after each \LOCAL round we compute \smash{$\max_{v \in V} \max \big(\tfrac{k}{N(v,d)\gamma}, d\big) = \max \big(\tfrac{k}{N(d)\gamma}, d\big) $} which takes $\tilO(1)$ rounds using the global network \NCC (see \cite{Augustine2019}).\footnote{Aggregation (e.g., computing of a maximum) of one value per node in $\tilO(1)$ in \NCC is not very hard, in particular if we are allowed to use randomization, \cite{Augustine2019} also offers a deterministic solution.} 	
	
	Since we alternate between a round of \LOCAL and a $\tilO(1)$ round aggregation via \NCC , the overall running time is $\tilO(d')$, where $d'$ is the depth to which we have to explore locally such that $$\max \big(\tfrac{k}{N(d')\gamma}, d'\big) = \min_{d\in[D_G]}\max \big(\tfrac{k}{N(d)\gamma}, d\big) = \NQ(G,k,\gamma),$$ 
	so it remains to quantify $d'$.
	As in our  previous proof we use that $N(d)$ is strictly increasing in $d$, thus \smash{$\tfrac{k}{N(d)\gamma}$} decreases strictly in $d$. Therefore, if after some $(d'+1)$-th round (assuming $d' < D_G$) in \LOCAL and the subsequent aggregation the nodes observe the condition
	$$\max \big(\tfrac{k}{N(d')\gamma}, d'\big) \leq  \max \big(\tfrac{k}{N(d'+1)\gamma}, d'\p1\big)$$
	then this term does not get any smaller for any $d > d'$ anymore (or $d' = D_G$) thus $\NQ(G,k,\gamma) = \max \big(\tfrac{k}{N(d')\gamma}, d'\big)$. Due to Lemma \ref{lem:neigh_qual_min_neigh_size} we have $d'  \in \bigO(\NQ(G,k,\gamma))$.
\end{proof}


\section{Neighborhood Quality - Upper Bound for \boldmath $(k,\!\ell)$-SP}
\label{sec:upper_bound}

In this section we give an approximation algorithm with stretch $(1\p\eps)$ for  the $(k,\ell)$-SP problem for certain parameters of $k,\ell$ that takes $\tilO\big(\NQ(G,k,\gamma)\big)$ rounds for any graph $G=(V,E)$, which is competitive with the best algorithm optimized for $G$ up to $\tilT(1)$ factors (however the latter  will only become apparent in Section \ref{sec:lower_bound}, where we show a corresponding lower bound that even holds for the $(k,1)$-SP problem). In particular, our solution works for $\ell \in \tilO(1)$ arbitrary sources and up to $\ell \leq \NQpar$ i.i.d.\ random sources.

To prove this, we draw on a recent algorithm \cite{Schneider2023, Schneider2023a}, which in turn employs techniques by \cite{Rozhon2022}, to solve the SSSP problem in just $\tilO(1)$ rounds in the \hybrid model. We express their result in the following Lemma.

\begin{lemma}[see \cite{Schneider2023a}]
	\label{lem:k-ssp}
	A $(1 \p \eps)$-approximation of SSSP can be computed in $\tilO(1/\eps^2)$ rounds in the \hybrid model, w.h.p.
	Furthermore, the $\ell$-SSP problem can be solved in $\tilO\big(\!\sqrt {\ell/\gamma}\cdot\tfrac{1}{\eps^2}\big)$ rounds (for $\eps > 0$) w.h.p.\ in the \hybridpar{\infty}{\gamma} model with stretch $3 \p \eps$.
\end{lemma}

We stress that Lemma \ref{lem:k-ssp} by itself does not give a universally tight solution, since the corresponding parameter $\NQ(G,k,\ell)$ for $\ell$ could be very small, i.e., $\NQ(G,k,\ell) \ll \sqrt{\ell}$ (in case nodes have quite large neighborhoods in relative proximity). However, the result above can be used for a universally tight solution of the $(k,\ell)$-SP problem as follows. 

We approximate the $\ell$-SSP problem where the $\ell \leq \NQpar$ target nodes take the role of source nodes. In particular, by Lemma \ref{lem:k-ssp}, all $k$ source nodes learn their distance to each of the $\ell$ targets. Thus we would be able to solve the $(k,\ell)$-SP problem if the $k$ sources can communicate the corresponding distances to the $\ell$ target nodes.

This corresponds to a message routing problem where all $k$ source nodes have to deliver a message to each of the $\ell$ targets. In Section \ref{sec:adaptive_helper_sets} we show that this message routing problem can be solved in  $\tilO\big(\NQpar\big)$ rounds by gearing some of the techniques of \cite{Kuhn2022} (algorithms for ``token routing'') for the parameter \NQpar by computing a structure of adaptive helper sets (Section \ref{sec:adaptive-helper-sets}).
Here, two bottlenecks that prohibit a larger set of target nodes emerge. 

The first bottleneck is the number $k$ of messages that each target node has to receive. This is where we have to leverage the neighborhood quality $\NQpar$ and the random selection of target nodes. The neighborhood quality gives each target node access to enough neighbors that allows it to effectively receive all those messages. The random selection of target nodes implies that these are well spread over in the local network, thus none has to share its neighbors with too many others w.h.p.

The second bottleneck is that source nodes are arbitrary and can be locally highly concentrated. This means they cannot meaningfully rely on other nodes to help them to send their messages, thus the number of messages they have to send (and thus $\ell$) needs to be bounded. 

\subsection{Adaptive Helper Sets}
\label{sec:adaptive-helper-sets}

To work around those bottlenecks and solve this routing problem we adapt the structure of so called ``helper sets''.  that was first employed in \cite{Kuhn2020} and our definition accommodates the neighborhood quality of a graph $G$.

\begin{definition}[cf.\ \cite{Kuhn2020}]
	\label{def:helpers}
	Let $G = (V,E)$, $k \in [n]$ and let $W \subseteq V$ either be a set of $\ell \leq \NQpar$ i.i.d.\ random nodes or $\tilO(1)$ arbitrary nodes.
	A family $\{H_w \subseteq V \mid w \in W\}$ of \emph{adaptive helper sets} has the following properties. 
	
 	(1) Each $H_w$ has size $\tilT\big(k/\gamma\NQpar\big)$. 
 	
 	(2) For all $u \!\in\! H_w\!: hop(w,u) \!\in\!  \tilO\big(\NQpar\big)$. 
 	
 	(3) Each node is part of at most $\tilO(1)$ sets $H_w$.
\end{definition}

The main difference to the helper sets by \cite{Kuhn2020} is that we obtain a stronger guarantee on their size, by leveraging the graph parameter \NQpar. We prove that we can compute such a helper set in the following lemma.

\begin{lemma}[see \cite{Kuhn2020}]
	\label{lem:helpers}
	A family of adaptive helper sets $\{H_w \subseteq V \mid w \in W \}$ as in Definition \ref{def:helpers} can be computed w.h.p.\ in $\tilO(\NQ(G,k,\gamma))$ rounds in \hybridpar{\infty}{\gamma}.
\end{lemma}

\begin{proof}
	First, nodes collaborate to compute the parameter $\NQpar$ in $\tilO\big(\NQpar\big)$ rounds, see Lemma \ref{lem:compute_neigh_qual}. The next step is to compute a so called $(\alpha, \beta)$-ruling set $R \subseteq V$ for some specific parameters $\alpha, \beta$. Such a set fulfills the property that any two nodes in $R$ have hop distance at least $\alpha$ and at the same time each node in $v \in V$ has a ``ruler'' $r \in R$ within $\beta$ hops.  If $\beta$ is by a logarithmic factor larger than $\alpha$, then $R$ can be computed in $\tilO(\alpha)$ rounds using only the local network (\cite{Awerbuch1989} achieves this in the \LOCAL model, whereas \cite{Kuhn2018a} shows that \CONGEST suffices). 	
	
	 We specifically compute a $(\alpha , \beta)$-ruling set with parameters $\alpha = 2\NQpar\p 1$, $\beta = 2\NQpar\lceil \log n \rceil$ in $\tilO\big(\NQpar\big)$ rounds using \cite{Awerbuch1989, Kuhn2018a}. The next step to obtain the adaptive helper sets is to compute a clustering $C_r \subseteq V, r \in R$ based on this ruling set, which partitions the nodes into clusters by each node $v$ joining the cluster $C_r$ of the ruler $r$ that has smallest hop distance to $v$ (break ties arbitrarily).  
	 
	 In a sense, such a clustering locally partitions $G$ into ``areas of responsibility'', and allows us to assign each node $w \in W$ a helper set $H_w$ from the cluster $C_r$ that $w$ is located in, in a fair way such that no node in $C_r$ has to help too many nodes from $W$.
	 For this to work, the nodes in $C_r$ first learn the sets $C_r$ and $C_r \cap W$. Note that the diameter of $C_r$ is in $\tilO\big(\NQpar\big)$, because nodes can always join a cluster that is at distance at most $\beta$ thus we can compute the required information in the same number of rounds using the local network. Then, for each node $w \in C_r \cap W$, each $v \in C_r$ joins the helper set $H_w$ with probability 
	 $$q = \min\big(\tfrac{k}{\gamma\NQpar} \cdot \tfrac{1}{|C_r|} \cdot  8c \ln n, 1\big)$$ 
	 It is clear that this assignment of helper sets fulfills property (2) of Definition \ref{def:helpers} due to the diameter of $C_r$. Let us now look at the size of each helper set. In case $q=1$ this means that each node in the cluster $C_r$ that $w$ is located in is drafted into $H_w$. Furthermore, within $(\alpha \m 1)/2 = \NQpar$ hops all nodes necessarily belong to $C_r$ by definition of the ruling set. Now we have to exploit the neighborhood quality of $r$. Using Lemma \ref{lem:neigh_qual_min_neigh_size} we have $$|H_w| = |C_r| \geq N(d') \geq k/(\NQpar\p 1)\gamma \in \tilT\big(k/\gamma\NQpar\big)$$ 
	 In case $q<1$, the expected size of $H_w$ is \smash{$\E(|H_w|) = \tfrac{k}{\gamma\NQpar} \cdot  8c \ln n$}, which is chosen such that we can conveniently apply a Chernoff bound (given in Lemma \ref{lem:chernoffbound} for completeness) to lower bound the size of $H_w$:
	 \[
	 	\Pr\Big(|H(X)| \leq  (1 - \tfrac{1}{2}) \E(|H_w|)\Big) \leq \exp \Big(\!-\!\tfrac{\E(|H_w|)}{8}\Big) \leq \exp \Big(\!-\!\tfrac{8c \ln n}{8}\Big) = \tfrac{1}{n^c}. 
	 \]
	 Therefore, the size of $H_w$ is at least $\E(|H_w|) /2 \in \tilO\big(k/\gamma\NQpar\big)$ w.h.p. 
	 
	 The third property is clear if $|W| \in \tilO(1)$. Else, we use that since $W$ consists of $\ell$ nodes chosen i.i.d.\ randomly from all $n$ nodes, thus the expected size of $C_r \cap W$ is 
	 $$\E\big(|C_r \cap W|\big) = |C_r| \cdot \tfrac{\ell}{n} \leq |C_r| \cdot  \tfrac{\gamma \NQpar}{k}.$$
	 Let $J_v = |\{w \in W \mid v \in H_w\}|$ be the random number of helper sets that a node $v \in C_r$ joins. In expectation these are at most $\E(J_v) = q \cdot  \E\big(|C_r \cap W|\big) \leq 8 c \ln n$ nodes. Again, we employ a Chernoff bound (Lemma \ref{lem:chernoffbound}) to show $J_v \in \tilO(1)$ w.h.p.:
	 \[
	 	\Pr\big(J_v \geq  (1+1) 8 c \ln n\big) \leq \exp \big(\!-\!\tfrac{8c \ln n}{3}\big) \leq  \tfrac{1}{n^c}. 
	 \]
	 Finally, we have to mention that the number of events that must occur simultaneously w.h.p.\ $1-\frac{1}{n^c}$ is polynomial in $n$. The constant $c$ can be adapted such that \textit{all} of them occur w.h.p.\ using a union bound (Lemma \ref{lem:unionbound}).
\end{proof}

\subsection{Solving the Routing Problem in \smash{$\boldmath \tilO\big(\NQpar\big)$} rounds}
\label{sec:adaptive_helper_sets}

Now, we are ready to efficiently solve the routing problem where each source has to send one distance label to each target in \smash{$\tilO\big(\NQpar\big)$} rounds. Note that Theorem \ref{thm:routing} is potentially much better  than the \smash{$\tilT\big(\!\sqrt{k}\big)$} complexity on worst case graphs (cf. \cite{Kuhn2020, Schneider2023}), as shown out in Lemma \ref{lem:sqrt_k_upper_bound}.

\begin{theorem}
	\label{thm:routing}
	Given $k$ source nodes and $\ell \in \tilO(1)$ arbitrary targets or $\ell \leq \NQpar$ i.i.d.\ random target nodes, where each source has $\ell$ tokens of size $\bigO(\log n)$ bits, one for each target. All tokens can be routed to their intended targets in \smash{$\tilO\big(\NQpar\big)$} rounds w.h.p.\ in the \hybridpar{\infty}{\gamma} model.
\end{theorem}

\begin{proof}
	First we compute the parameter \NQpar in \smash{$\tilO\big(\NQpar\big)$} rounds w.h.p.\ according to Lemma \ref{lem:compute_neigh_qual}. Let $S$ be the set of $k$ sources and let $T$ be the set of $\ell$ targets. 
	Next, the nodes in $T$ broadcast their IDs, which can be accomplished w.h.p.\ in $\tilO\big(\!\sqrt \ell \big)$ rounds using the broadcast protocol by \cite{Augustine2020a} so all nodes can order $T = \{t_1, \dots, t_\ell\}$ by ascending IDs. We also aggregate the number of sources $k$ in $\tilO(1)$ rounds using the aggregation routine by \cite{Augustine2019}. 
	
	Then each $t \in T$ computes its adaptive helper set $H_t$ in \smash{$\tilO\big(\NQpar\big)$} rounds w.h.p., as shown in Lemma \ref{lem:helpers}. The idea for the remaining proof is to assign each sender $s \in S$ a helper $v \in H_t$ of each target node $t \in T$. If $s$ would know the helpers that are supposed to receive its tokens, then it could in theory send them its $\ell$ tokens sequentially. This assignment can theoretically be done such that no helper has too many tokens to receive. The caveat is that there are too many helpers (close to $k$) and sources ($k$) to make their IDs known to each other in  \smash{$\tilO\big(\NQpar\big)$} rounds, so this is out of the question.
	
	The required information to assign helpers of targets to senders can be condensed significantly by instead relaying tokens via a pseudo-random set of intermediate nodes (cf.\ \cite{Kuhn2020}) using a hash-function selected randomly from a suitable universal family $\calH$. For this to work it suffices to publish a small random seed using the global network.
	The computational and probabilistic aspects of  a suitable family $\calH$ from the literature (in particular the indpenece number of hash values) are summarized in Definition \ref{def:hashfunctions} and Lemma \ref{lem:hashfunctions} in Appendix \ref{apx:Pseudorandom}.
	
	We use a family $\calH$ consisting of functions $h : [k] \times [\ell] \to [n]$. That is, $h \in \calH$ maps a pair $(i, j), i \in [ k], j \in [\ell]$ to an intermediate node with ID $h(i,j) \in [n]$ over which the token from source $i$ to target $j$ will be relayed.
	Since the IDs and the order of the target nodes $T = \{t_1, \dots , t_\ell\}$ are publicly known, every node can relate index $j$ to $ID(t_j)$. The same is not true for the set of sources $S$, which we cannot publish as it is too large. 
	To this more cheaply, we define another family $\calG$ of functions $g : [n] \to [k]$ intended to relate IDs of the sources $S$ to indices $[k]$. In the following we assume the two random seeds to select $h \in  \calH, g \in  \calG$ are already known, so nodes can use $g,h$ (we determine the size of the two seeds afterwards).
	
	We are now ready to describe the transmission process. Here, nodes may assume multiple roles as senders, targets, helpers or intermediates and we are going to analyze and bound the workload of each node later. The rough idea is to first transmit all tokens from senders to intermediate nodes via the global network. Then helpers $H_{t_j}$ request the tokens for target node $t_j$ from the intermediate nodes. After the helpers obtained all tokens from intermediate nodes the target node $t_j$ collects all tokens from its helpers via the local network.
	
	In more detail, first $s \in S$ sends each token and its identifier $ID(s)$ to the corresponding intermediate node with ID $h\big(g(ID(s)), j\big)$, where $j \in [\ell]$ is the intended target. We throttle the number messages per round to batches of size at most \smash{$b \in \Omega\big(\gamma/{\log^2 n}\big)$}, to not exceed the global receive capacity $\gamma$ of intermediate nodes as we show later (note, $b\geq 1$ for some $\gamma \in \tilOm(1)$, cf., Definition \ref{def:hybrid}). After that, all tokens are located at intermediate nodes ready to be retrieved by the helpers $H_{t_j}$ of target nodes $t_j \in T$.	
	
	For this purpose, each target $t_j \in T$ creates $k$ tasks $(i, j), i \in [k]$, which it distributes evenly among its helpers using the local network. Subsequently, each helper $v \in H_{t_j}$ sends a request $(ID(v),i,j)$ to the node with ID $h(i, j)$, for each task $(i,j)$ it was assigned. Specifically, $v$ sends a batch of size $b$ of such requests per round.	 
	 On receiving a request $(ID(v),i,j)$ in the previous round, the node with ID $h\big(i, j\big)$ responds by sending all tokens with source $s \in S$ with $g(ID(s))=i$ and target $t_j$ to $v$. 
	 After all helpers received the responses for all of their requests in this way, the corresponding tokens adressed to target $t_j$ are now located at its helpers $H_{t_j}$ and the local network can be used to collect them.
	 
	 We show two properties that complete the proof. Firstly, we prove that the number of rounds it takes to transmit all messages over the local and global network, respectively, is at most \smash{$\tilO\big(\NQpar\big)$}. Secondly, we prove that the procedure correctly delivers all tokens addressed to any $t_j$. The main challenge here is to show that the number of global messages any node receives in any given round is bounded by $\gamma$, so that no messages are dropped and lost (see Definition \ref{def:hybrid}).
	 
	We start with the round complexity, which by the design of the algorithm can be split into rounds of only local or global communication. The only local communication that takes place is between target nodes $t_j \in T$ and their helper sets $H_{t_j}$ to assign tasks and subsequently collect the tokens. Since, by Lemma \ref{lem:helpers}, helpers $H_{t_j}$ are within \smash{$\tilO\big(\NQpar\big)$} hops of $t_j$ and since the local network is unrestricted in terms of capacity we can do this for all helper sets in \smash{$\tilO\big(\NQpar\big)$} rounds.
	
	On the global network, we look at the round complexity of nodes in their different roles. 	
	First, the source nodes $S$ send their $\ell$ tokens sequentially to intermediate nodes defined by the hash functions $h,g$, which takes at most $\tilO(\ell) \in \tilO(\NQpar)$ rounds. Consider a  node $v$ that is a member of at least one helper set $H_{t_j}$. In total, there are $|H_{t_j}| \in \tilT\big(k/\gamma\NQpar\big)$ helpers available to $t_j$ for distributing its $k$ tasks $(i,j)$. Thus $v$ is assigned at most $\lceil k/|H_{t_j}|\rceil \in \tilO(\gamma\NQpar)$ tasks. By Lemma \ref{lem:helpers}, $v$ has to help at most $\tilO(1)$ target nodes $t_j \in T$, thus the total assigned tasks by all $t_j \in T$, that $v$ is a helper of is still $\tilO(\gamma\NQpar)$. Since $v$ sends $\gamma/b \in \tilO(\gamma)$ requests per round, the time to send all requests it was tasked with takes $\tilO(\NQpar)$ rounds.

	Let us now look at some intermediate node $u$, i.e., a node that obtains at least one token by a sender $s$ with $g(ID(s))=i$, which $u$ has to forward to some helper $v$ of some target $t_j$, after receiving the according request $(ID(v),i,j)$. We bound the number $X_u = |\{(i,j) \mid u = h(i,j)\}|$ and the number $Y_i = |\{s \in S \mid i = g(ID(s))\}|$ for $i \in [k]$. The random value $X_u$ is the number pairs $(i,j)$ that $u$ is responsible for. The second is the number of ID's of some $s \in S$ for which an index in $i \in [k]$ is ``hit'' by the hash function $g$, i.e., $g(ID(s)) = i$.
	
	We assume that $h,g$ provide i.i.d.\ randomness with $\bigO\big(\ell \log n\big)$ independence. Then 	 bounding $X_u,Y_i$ corresponds to simple balls into bins problems. For $X_u$, we have $k \cdot \ell$ balls and $n\, (\geq k)$ bins. For $Y_i$ we have $|S| = k$ balls and $k$ bins (indices in $[k]$). We give a simple solution for these balls into bins problems in the separate Lemma \ref{lem:balls_in _bins} in Appendix \ref{apx:generalnotations}, which implies $X_u \in \bigO(\ell\log n)$ and $Y_i \in \bigO(\log n)$ w.h.p.
	We conclude that no intermediate node $u$ obtains more than $X_u \cdot Y_i  \in \tilO(\ell)$ requests w.h.p., which means that $u$ can forward all tokens which are requested within $\tilO(\ell)$ rounds w.h.p.
	
	Let us finally come to the correctness, by which we mean that no node receives more than $\gamma$ bits via the global network in any given round w.h.p., thus no message is dropped (cf., Definition \ref{def:hybrid}). Note that only intermediate nodes and helpers have to receive global messages in the algorithm described above. We start by showing the claim for intermediate nodes.
	
	Given that no source node sends more than $b$ tokens per round to intermediate nodes selected uniformly at random, then the number of received tokens per round is $\bigO(b)$ w.h.p., by Lemma \ref{lem:receiveBound}. However, intermediate nodes are selected according to some hash function $h(i,j)$, where $g\big(ID(s)\big)=i$. Assuming $h$ has independence at least $\Omega(b)$ the claim of by Lemma \ref{lem:receiveBound} is retained, given that for any two sources $s_1,s_2$ it is $g\big(ID(s_1)\big) \neq g\big(ID(s_2)\big)$, since then we still have unique keys. 
	
	That is unfortunately not the case, however, the number of sources $s \in S$ with the same value $i = g\big(ID(s)\big)$ can instead be bounded by $Y_i \in \bigO(\log n)$ w.h.p., that is, we can guarantee that no node receives more than $\bigO(b \cdot Y_i) = \bigO(b \cdot \log n) $ messages per round w.h.p. Given the maximum size of a message is $\bigO(\log n)$, the number of bits received per round is at most $\bigO(b \cdot \log^2 n)$. Adjusting the constant in the batch size $b \in \Omega(\gamma/\log^2 n)$ accordingly, no intermediate node receives more than $\gamma$ bits per round.	
	
	Finally we analyze the seeds for the required hash functions $h,g$ that need to be shared. We require $h,g$ to map values in sets that are at most as large as the ID space we are working with and these are required to provide i.i.d.\ randomness with independence $\bigO\big(\max(b, \ell \log n)\big)$. By Lemma \ref{lem:hashfunctions} these requirements can be met using seeds of size at most $\tilO\big(\max(\gamma, \ell )\big)$. Moreover, the seeds can be determined locally by a single node and then broadcast in $\tilO(\ell)$ rounds using the broadcast algorithm by \cite{Augustine2020a}.
\end{proof}

\subsection{Solving \boldmath $(k,\ell)$-SSP in $\tilO\big(\NQpar\big)$ rounds}

It remains to combine the tools derived in the previous two subsections with the fast solutions for $\ell$-SSP to prove the following theorem.

\begin{theorem}
	\label{thm:uni_upper_bound}
	Let $\eps > 0$ be an arbitrary constant. The $(k,\ell)$-SP problem can solved w.h.p.\ in the \hybridpar{\infty}{\gamma} model in $\tilO\big(\NQ(G,k,\gamma)\big)$ rounds with stretch $(1 \p \eps)$ for $\ell  \in \bigO(1)$  and with stretch $(3 \p \eps)$ for $\ell \leq \NQpar$  i.i.d.\ random target nodes. 
\end{theorem}

\begin{proof}
	The first step is to solve the $\ell$-SSP problem for the set of target nodes using Lemma \ref{lem:k-ssp}. Since we consider $\varepsilon$ constant and $\ell \leq \NQpar$, the round complexity for this step is at most $\tilO\big(\NQpar\big)$ and the stretch is $1 \p \eps$ or $3 \p \eps$ depending on the set of target nodes.	
	Afterwards, each source $s$ knows its distance to each target $t$, but we require this to be the case the other way around in order to solve the $(k,\ell)$-SSP problem. This is now a simple application of Theorem \ref{thm:routing}, where each source $s$ puts its distance to each target $t$ into a token, and all such tokens are delivered in $\tilO\big(\NQpar\big)$ rounds.
\end{proof}

\section{Neighborhood Quality - Lower Bound for \boldmath $(k,1)$-SP}
\label{sec:lower_bound}

We will show that the neighborhood quality parameter constitutes a lower bound for the $k$-SSP problem, even for the $(k,1)$-SP problem where only a single node must learn its distance to $k$ others and even if all nodes know $G$. We express this in the following theorem and dedicate the remainder of this section to its proof.

\begin{theorem}
	\label{thm:universal_lower_bound}
	Assume each node obtains the local graph $G=(V,E)$ as part of its input. Then a randomized algorithm that computes a  polynomial approximation of the $(k,1)$-SP problem (see Def.\ \ref{def:kSSP}) on $G$  with constant probability $p$ in the \hybridpar{\infty}{\gamma} model still takes \smash{$\Omega\big(\NQ(G,k,\gamma)\big)$} rounds.
\end{theorem}


Let us start with some of the required technical claims.

\subsection{Technical Properties}

The first lemma on the path to prove Theorem \ref{thm:universal_lower_bound} shows that the nodes $V'$ outside a radius $d< D_G$ of a node $v$ can be partitioned into non-trivial sets $V_1,V_2$ of size roughly $\Theta(n')$ each ($n' := |V'|$, assuming $n' > 1$) and an assignment of weights can be given, such that distances from $v$ to nodes in $V_2$ are much larger than to nodes in $V_1$.
In the proof we make an important case distinction on two fundamentally different layouts of $G$, where $D_{G}$ is either large or relatively small. The proof gives a construction of such a partition with large difference in distance from $v$ for either case.

\begin{lemma}
	\label{lem:node_partition}
	Let $G = (V,E)$, $v \in V$ and $d < D_G$. Let $V' := V \setminus \calB(v,d)$ and $n' := |V'| \geq 8$. Then there is a partition $V_1 \cup V_2 = V', V_1 \cap V_2 = \emptyset$ such that the following holds.  We have $|V_1|, |V_2| \geq n' /8$ and there is a weight assignment $w:E \to [W]$, s.t.\ for any two nodes $v_1 \in V_1,v_2 \in V_2$ we have $d(v,v_1) \leq p(n) d(v,v_2)$ for any (fixed) polynomial $p(n) \geq n$.
\end{lemma}

\begin{proof}	
	Consider a breadth first search (BFS) tree $T_v$ with root $v$.  We aim to identify at least two sets $\calT_1,\calT_2$ of disjoint sub-trees in $T_v$, each of which has a cumulative number of nodes at least $n'/8$, together with a set of edges $E'$ such that any shortest path in $T_v$ from $v$ to some node $w \in T \in \calT_2$ has to use an edge of $E'$ and no shortest path from $v$ to any $u \in T \in \calT_1$ contains an edge of $E'$.
	
	Let $T_1, \dots, T_\ell$ be the set of BFS sub-trees induced by $T_v$ on $V'$. Due to the definition of $V' = V \setminus \calB(v,d)$ the trees $T_1, \dots, T_\ell$ are indeed BFS trees on $V'$ (which might have different connected components) rooted at the respective common ancestor in $T_v$. In the following we use $|T|$ to denote the number of nodes in a tree $T$.
	
	If $|T_i| \leq n'/2$ for all $i\in [\ell]$, then we can select two disjoint subsets $\calT_1,\calT_2 \subseteq \{T_1, \dots, T_\ell\}$ with $n'/4 \leq \sum_{T \in \calT_j} |T| \leq n'/2$ as follows. A large tree with $n'/4 \leq |T_i| \leq n'/2$ will form its own set $\calT_j$ (of course we select at most two). Unless two large trees already give us the desired $\calT_1,\calT_2$, we can combine the remaining $|T_i| \leq n'/4$ into one or two of the sets $\calT_1,\calT_2$. We define $E'$ as the parent edges (w.r.t.\ $T_v$) of the root nodes of the trees in $\calT_2$.
	
	Otherwise, consider the tree $T_i$ with $n_i :=|T_i| > n'/2$. We consider a so called \textit{splitting node} $x$ of $T_i$ the removal of which splits $T_i$ into sub-trees of size at most $n_i/2$ nodes. Note that such a splitting node always exists, but for the sake of reading flow we show its existence in the separate Lemma \ref{lem:split_node}. After removal of $x$, let $T'$ be the sub-tree of $T_i$ containing the parent of $x$ and let $T_1', \dots, T_m'$ be the sub-tress containing the children of $x$ in $T_i$. 
	
	Now we can construct of $\calT_1, \calT_2$ with $\sum_{T \in \calT_j} |T| \geq n'/8$ from the sub trees of $T_i$ which is quite similar to before. If there is a large tree $T \in \{T', T_1', \dots, T_m'\}$  with $|T| \geq n_i/4 \geq n'/8$ then we choose it as one of the two sets $\calT_1,\calT_2$. Note that if this is the case for $T'$ we are always going to define it as $\calT_1$.
	
	If this does not yet give us both sets $\calT_1,\calT_2$ we consider the trees $T_1', \dots, T_m'$. If we are still lacking \textit{both} tree sets then it was $|T'| < n_i/4$ therefore $T_1', \dots, T_m'$ must contain at least $3n_i/4$ nodes. This means that we can choose disjoint $\calT_1,\calT_2 \subseteq \{T_1', \dots, T_m'\}$ each containing at least $n_i/4 \geq n'/8$ nodes. If just one tree set is missing (w.l.o.g.\ $\calT_2$), then $\calT_2 := \{T_1', \dots, T_m'\} \setminus \calT_1$ must still contain trees with at least $n_i/4 \geq n'/8$ nodes. 	
	Again we set $E'$ as the parent edges (w.r.t.\ $T_v$) of the root nodes of the trees in $\calT_2$. 
		
	Note that the way we constructed $E'$ in both cases, the shortest path in $T_v$ from $v$ to some $u \in T \in \calT_2$ has  to cross an edge in $E'$, which is never the case for any $w \in T \in \calT_1$! 	
	It remains to assign the weights $w : E \to [W]$ suitably.
	
	All edges that are not part of $T_v$ obtain weight $n \cdot p(n) \p n$ (which is the maximum polynomial weight $W$ we use). 	
	All edges in $E' $ obtain weight $n \cdot p(n)$. All remaining edges in $T_v$ obtain weight $1$. This gives us the following. (1): all distances in $T_v$ are strictly smaller than $n \cdot p(n) \p n$, thus edges outside of $T_v$ are not viable. (2): the distance to any $u \in T \in \calT_1$ is at most $n$. (3): the distance to any $w \in T \in \calT_2$ is at least $n \cdot p(n)$. 
	
	We set $V_1,V_2$ as the set of nodes of the trees in $\calT_1,\calT_2$. Then the claim of the lemma can be directly inferred from (1)-(3).
\end{proof}

For completeness, we show the existence of a splitting node used in the proof above.

\begin{lemma}
	\label{lem:split_node}
	Let $T = (V_T,E_T)$ be a tree on $n$ nodes rooted at some $v \in V_T$. There exists a splitting node $x \in V_T$ whose removal splits $T$ into sub-trees with at most $n/2$ nodes each.
\end{lemma}

\begin{proof}
	For any $u\in V_T$ let $t(u)$ be the number of nodes in the sub-tree rooted at $u$ (from the vantage point of the root $v$) and let $p(u) := n - t(u)$. Consider a path $P$ from $v$ to some leaf that always follows a node that maximizes  $t(u)$.  For $u$ on $P$ let $s(u)$ be the successor of $u$ on $P$ (which starts at $v$). 
	
	The lemma is proven if we find a node $x$ that fulfills two conditions. First, $p(x) \leq n/2$ and second $t(c) \leq n/2$ for every child node $c \in V_T$ of $x$. Let $x$ be the node on $P$ with $p(x) \leq n/2$ and $p(s(x)) \geq n/2$. The former satisfies the first condition. The latter implies $t(s(x)) \leq n - p(s(x)) \leq n/2$. Since $t(s(x))$ is largest among all children $c$ of $x$ in $T$, we have  $t(s(c))\leq n/2$, which satisfies the second condition.
\end{proof}

Next, we show that there is a node $v$ that has a large number of nodes outside of a radius of $d _v\m 1$ of $v$.

\begin{lemma}
	\label{lem:nodes_outside_neighborhood}
	 Let $G=(V,E)$ be a graph, $k \in[n]$ and let $v$ be the node that minimizes  $\NQ(v)$ and $d_v \geq 1$ (see Definition \ref{def:neigh_qual_node}), then $N(v,d_v \m 1) \geq k/2$.
\end{lemma}

\begin{proof}
	The claim is clearly true in case of $d_v = 1$. In the case $d_v > 1$ assume $N(v,d_v \m 1) > k/\gamma$ (for a contradiction).  Then we would have
	\[
	\max \Big(\tfrac{k}{N(v,d_v-1)\gamma}, d_v-1\Big) \leq  \max (1, d_v \m1) < \max \Big(\tfrac{k}{N(v,d_v)\gamma}, d_v\Big).
	\]
	This implies that $d_v \m 1$ minimizes $\NQ(v)$, a contradiction. 	
	As a consequence we have 
	$$\smash{n - N(v,d_v \m 1) \geq n - k/\gamma \geq n - n/\gamma \stackrel{\gamma \geq 2}{\geq} n/2 \geq k/2.}$$
	Note that we assume $\gamma \in \tilOm(1)$, thus $\gamma \geq 2$.
\end{proof}

Finally, we show the following technical lemma which will be important later on. It is based on the Definition \ref{def:neigh_qual_node} of $\NQ(v)$ and $d_v$.

\begin{lemma}
	\label{lem:ineq_neigh_qual}
	Let $G=(V,E)$ be a graph and $k \in[n]$. Let $v \in V$ and $1 \leq d_v < D_G$. Then $$\smash{\NQ(v) = \max \Big(\tfrac{k}{N(v,d_v)\gamma}, d_v\Big) \leq  \min \Big(\tfrac{k}{N(v,d_v-1)\gamma}, d_v\Big)+1}$$
\end{lemma}

\begin{proof}
	Since $d_v$ minimizes the outer minimum of $\NQ(v)$, we have
	\[
	\NQ(v) = \max \big(\tfrac{k}{N(v,d_v)\gamma}, d_v\big) \leq 
	\begin{cases}
		\max \big(\tfrac{k}{N(v,d_v+1)\gamma}, d_v+1\big) \quad \text{(a)} \\
		\max \big(\tfrac{k}{N(v,d_v-1)\gamma}, d_v-1\big)  \quad \text{(b)}
	\end{cases}
	\]
	Since $1/N(v,d)$ is monotonically decreasing in $d$ the inequalities (a),(b) imply
	\begin{align*}
		&\text{from (a): }\max \big(\tfrac{k}{N(v,d_v)\gamma}, d_v\big) \leq d_v+1  \text{ \textbf{ and } }\\ 
		&\text{from (b): }\max \big(\tfrac{k}{N(v,d_v)\gamma}, d_v\big) \leq \tfrac{k}{N(v,d_v-1)\gamma} 
	\end{align*}
	Combining the two inequalities above results in the claim.
\end{proof}

\subsection{Proof of Theorem \ref{thm:universal_lower_bound}}

Our goal for the remaining proof of the lower bound is to show that a node $v \in V$ with a relatively small neighborhood within a radius  (which is described by the neighborhood quality), has to learn a lot of information from outside that radius to learn its distance to all $k$ sources, even if $v$ is completely aware of the topology of $G$. 

Since we are not allowed to modify edges of $G$, the remaining variables to ``create'' information that is unknown to $v$ are the distances in $G$ and the placement of the set of sources. Nevertheless, we will show a reduction from a so called node communication problem \cite{Kuhn2022} to the $(k,1)$-SP problem (even when given $G$ as part of its input) and we will see that this takes $\Omega\big(\NQ(G,k,\gamma)\big)$ rounds.

The node communication problem was introduced as an abstraction for the problem where a part of the network has to learn some information  that only another (distant) part of the network knows. To give a rough description of the problem, we have two sets of nodes $A$ and $B$ at hop distance at least $h$, where nodes in $B$ ``collectively know'' the state of some random variable $X$ and need to communicate it to $A$, which has no information about $X$. 

It was shown \cite{Kuhn2022,Schneider2023} that a \hybrid algorithm that solves the node communication problem induces a transcript of the global communication that happened during its execution that can be used to obtain a uniquely decodable code for $X$, which implies a lower bound for the size of the transcript by Shannons source coding theorem. 

This, in turn, implies a lower bound for the number  of rounds it must have taken to solve the two party communication problem (the formal definition of the node communication problem is given in Appendix 
\ref{sec:node_comm}). We will use the following, slightly simplified version of the claim from \cite{Schneider2023} (which works even given that $G$ is known).

\begin{lemma}[cf. \cite{Kuhn2022, Schneider2023}]
	\label{lem:lower_bound_node_comm}
	Let $A,B$ be disjoint node sets and let $h \leq \hop(A,B)$ and $N := |\calB(A,h-1)|$.
	Any algorithm that solves the {node communication problem} (formally given in Def.\ \ref{def:node_comm_problem}) with $A,B$ in the \hybridpar{\infty}{\gamma} model with success probability at least $p$ for any r.v.\ $X$, takes at least \smash{$\min\!\big(\frac{pH(X) -1}{N \cdot \gamma}, \frac{h}{2} \m 1\big)$} rounds in expectation.
\end{lemma}

\begin{proof}[Proof of Theorem \ref{thm:universal_lower_bound}]
	Let $G=(V,E)$ be an arbitrary graph. Fix the target node $v := \argmax_{v \in V} \NQ(v)$  for which the $(k,1)$-SP problem needs to be solved on $G$. Furthermore, let \smash{$d_v := \argmin_{d \in [D_G]} \max \big(\tfrac{k}{N(v,d)\gamma}, d\big)$} (cf.\ Definition \ref{def:neigh_qual_node}).
	
	Note that we assume $d_v \geq 1$, which is the case unless $k \leq \gamma$. If  indeed $k \leq \gamma$ and $d_v = 0$ would be the case, then $\NQ(G,k,\gamma) \in \Theta(1)$ thus the claim of the theorem becomes trivial. Due to Lemma \ref{lem:nodes_outside_neighborhood} we have $N(v,d_v \m 1) \geq k/2$.
	
	Now let us invoke Lemma \ref{lem:node_partition} on $V' := V \setminus \calB(v,d_v\m1)$, which gives us two node sets $V_1,V_2$ with $V_1,V_2 \geq |V'|/2 \geq k/16 =: k'$ and a weight assignment $w: E \to [W]$ such that distances from $v$ to a node in $V_2$ is by a factor $q(n)$ longer than to any node in $V_1$ for any fixed polynomial $q$.
	
	Consider a random bit string $X = (x_i)_{i \in [k']}, x_i \in \{0,1\}$ of length $k'$. We enumerate $k'$ nodes from each set $v_{1,i}\in V_1, v_{2,i} \in V_2$ with indices  $i \in [k']$. Then, a set of $k'$ source nodes is selected from $V_1 \cup V_2$ according to $X$ as follows:  if $x_i =0$ we tag $v_{i,1}$ as a source, else we tag $v_{i,2}$ as a source.
	
	Assume that $v$ is given the knowledge of the weight function $w$, the number $k'$, the IDs of the nodes in $V_1,V_2$ and which index $i$ is associated with which pair of IDs in $V_1,V_2$ (which can only make the problem simpler). This is in addition to the knowledge of the topology of $G$, which $v$ is assumed to have anyway due to this being a universal lower bound.
	
	Note that $v$ does not have any initial knowledge whether $v_{1,i}$ or $v_{2,i}$ was selected as source and thus has no knowledge about $X$. Or to express this in terms of information theory: the initial state of $v$ is independent from $X$. To finish setting up the node communication problem let us define $A :=\{v\}$ and $B := V_1 \cup V_2$. Since $B \subseteq V \setminus \calB(v,d_v\m 1)$ we have $\hop(A,B) \geq d_v-1$. We choose $h = d_v-1 \leq \hop(A,B)$. Finally, we have that $N = N(v,d_v-1)$.
	
	Presume that we solve the $(k,1)$-SP problem with approximation factor at most $q(n)-1$ (where the polynomial $q$ can be chosen freely) with probability at least $p$. Thus the node $v$ learns its distance to each source $s_i$ up to factor $q(n)-1$. Since $v$ knows that its distances to nodes in $V_2$ are a factor of $q(n)$ larger than to those in $V_1$, this approximation lets $v$ determine which nodes were selected as source and thus $v$ will be able to determine the state of $X$ with probability at least $p$.
	
	We have therefore solved the node communication problem with the following parameters. The node $v$ has learned,  with probability of success $p$, the state a random variable $X$ which it previously had zero knowledge of, which has Shannon entropy $H(X) = k' \in \Theta(k)$ (see Def. \ref{def:entropy}) .
	Plugging those parameters into Lemma \ref{lem:lower_bound_node_comm}, the number of rounds for solving $(k,1)$-SP must have been at least \smash{$\min\!\big(a \cdot \frac{k}{N(v,d_v-1) \cdot \gamma}, b \cdot d_v  \big)$}, for fixed constants $0 < a,b < 1$. 
	
	It remains to relate this lower bound for the number of rounds to solve $(k,1)$-SP to the neighborhood quality $\NQ(G,k,\gamma)$.	
	\begin{align*}
			\min \big(a \cdot \tfrac{k}{N(v,d_v-1) \cdot \gamma}, b \cdot d_v  \big) & \geq \min(a,b) \cdot \min\big(\tfrac{k}{N(v,d_v-1) \cdot \gamma}, d_v  \big)\phantom{\Big(\Big)}\\
			& \geq \min(a,b) \cdot \Big(\max\big(\tfrac{k}{N(v,d_v) \cdot \gamma}, d_v  \big) -1\Big) \tag*{\text{\textit{Lemma} } \ref{lem:ineq_neigh_qual}}\\
			& = \min(a,b) \cdot \Big(\smash{\min\limits_{d \in [D_G]}} \max\big(\tfrac{k}{N(d) \cdot \gamma}, d  \big) -1\Big) \tag*{\text{\textit{Definition}} \ref{def:neigh_qual_node}}\\
			& = \min(a,b) \cdot  \big(\NQ(G,k,\gamma) -1\big) \phantom{\Big(\Big)} \tag*{\text{\textit{Definition} \ref{def:neigh_qual}}}\\
			& \in \Omega\big(\NQ(G,k,\gamma)\big).\phantom{\Big(\Big)} \tag*{\qedhere}
	\end{align*}	
\end{proof}

\phil{For SODA submission: Show that neighborhood quality is not a universal lower bound for  unweighted (k,1)-SP with unknown sources by giving a graph with only only one large ``distance class'' where an optimal algorithm is much faster (only figure out the few sources outside that large distance class). The same is obviously true if there are weights but they are known to all nodes (since the unweighted case is a special case which can be exploited as explained above). For the weighted (k,1)-SP with known topology *and* sources but unknown we can also get a faster ($1 \p \eps$) solution in $\tilO(1)$ rounds if the topology of the graph is path. Then can divide stretches of that path into $\log_{1 \p \eps}$ distance classes, where a run of SSSP suffices to learn where the boundaries between the distance classes are.}

\bibliographystyle{acm}
\bibliography{ref/ref}

\appendix

\section{Basic Probabilistic Concepts}
\label{apx:generalnotations}

We introduce a few basic probabilistic concepts that we are going to use thorughout this article.

\begin{lemma}[Chernoff Bound]
	\label{lem:chernoffbound}
	We use the following forms of Chernoff bounds in our proofs:
	$$\mathbb{P}\big(X > (1 \!+\! \delta) \mu_H\big) \leq \exp\Big(\!-\!\frac{\delta\mu_H}{3}\Big),$$
	with $X = \sum_{i=1}^n X_i$ for i.i.d.\ random variables $X_i \in \{0,1\}$ and $\mathbb{E}(X) \leq \mu_H$ and $\delta \geq 1$. Similarly, for $\mathbb{E}(X) \geq \mu_L$ and $0 \leq \delta \leq 1$ we have
	$$\mathbb{P}\big(X < (1 \!-\! \delta) \mu_L\big) \leq \exp\Big(\!-\!\frac{\delta^2\mu_L}{2}\Big).$$
\end{lemma}

\begin{remark}
	\label{rem:chernoffbound}
	Note that the first inequality even holds if we have $k$-wise independence among the random variables $X_i$ for $k \geq \lceil \mu_H \delta \rceil$ (c.f., \cite{Schmidt1995} Theorem 2, note that a substitution $\mu_H := (1 \!+\! \eps)\E(X)$ generalizes the result for any $\mu_H \geq \E(X)$).
\end{remark}

\begin{lemma}[Union Bound]
	\label{lem:unionbound}
	Let $E_1, \ldots ,E_k$ be events, each taking place w.h.p. If $k \leq p(n)$ for a polynomial $p$, then $E \coloneqq \bigcap_{i=1}^{k} E_i$ also takes place w.h.p.
\end{lemma}

\begin{proof}
	Let $d \coloneqq \deg(p)\!+\!1$. Then there is an $n_0 \geq 0$ such that $p(n) \leq n^d$ for all $n \geq n_0$. Let $n_1, \ldots , n_k \in \mathbb{N}$ such that for all $i \in \{1, \ldots, k\}$ we have $\mathbb{P}(\overline{E_i}) \leq \tfrac{1}{n^c}$ for some (yet unspecified) $c > 0$.
	With Boole's Inequality (union bound):
	\begin{align*}
		\mathbb{P}\big(\overline{E}\big) \!= \mathbb{P}\Big(\bigcup_{i=1}^{k} \overline{E_i} \Big) \leq \sum_{i=1}^{k} \mathbb{P}(\overline{E_i}) \leq \sum_{i=1}^{k} \!\frac{1}{n^c} \leq \frac{p(n)}{n^{c}} \leq \frac{1}{n^{c-d}}
	\end{align*}
	for all $n \geq n_0' \coloneqq \max(n_0, \ldots ,n_k)$. Let $c' > 0$ be arbitrary. We choose $c \geq c' \!\!+\! d$. Then $\mathbb{P}\big(\overline{E}\big) \leq \frac{1}{n^{c'}}$ for all $n \geq n_0'$.
\end{proof}

\begin{remark}
	If a constant number of events is involved we typically use the above lemma without explicitly mentioning it. It is possible to use the lemma in a nested fashion as long as the number of applications is polynomial in $n$.
\end{remark}

An application of the lemmas above is the classic balls into bins problem that on an abstract level describes the maximum number of tasks per node of a i.i.d.\ random assignment of such tasks to nodes (i.e., balls into bins). We show the following bound.

\begin{lemma}
	\label{lem:balls_in _bins}
	Given $\ell$ bins and \emph{at most} $m$ balls with $m \in  \Omega(\ell \log n)$. We assign each ball to a bin i.i.d.\ and $k$-wise independently for some $k \in \Omega\big(\tfrac{m}{\ell}\big)$. Then no bin contains more than $a \cdot \tfrac{m}{\ell}$ balls w.h.p.\ for some constant $a>1$ and sufficiently large $n$.
\end{lemma}

\begin{proof}
	Let $X_{v}$ be the number of balls that bin $v$ contains. We have $\mathbb{E}\big(X_{v}\big) \leq \tfrac{m}{\ell}$. By the premise, we have $m \geq \xi \ell \ln n$ for some constant $\xi$ and large enough $n$. Let $c>0$ be an arbitrary constant.  We choose \smash{$a \geq (1 \!+\! \frac{3c}{\xi})$}. Then a Chernoff bound (Lemma \ref{lem:chernoffbound}, combined with Remark \ref{rem:chernoffbound}) yields
	\[\mathbb{P}\Big(X_{v} \!>\! a \cdot \tfrac{m}{\ell} \Big) \leq \mathbb{P}\Big(X_{v} \!>\! (1 \!+\! \tfrac{3c}{\xi})\tfrac{m}{\ell} \Big) \leq \exp\Big(\!-\! \frac{3\xi c \ln n}{3\xi}\Big) =  \frac{1}{n^{c}}.\]
	By the union bound given in Lemma~\ref{lem:unionbound} the event \smash{$\bigcap_{v \in V} X_{v} \!\leq\! \gamma $} also takes place w.h.p.
\end{proof}

One application of the lemma for the balls into bins problem that we use in this work is given by following lemma.

\begin{lemma}
	\label{lem:receiveBound}
	Presume some algorithm operating on a graph with $n$ nodes takes at most $p(n)$ rounds for some polynomial $p$. Let $m \in \Omega(\log n)$. Presume that each round, every node sends at most $m$ messages via the global network to target nodes in $V$ picked uniformly at random and $k$-wise independently for some $k \in \Omega(m)$. Then we can choose constant $a >1$ such that every node receives at most $a \cdot m$ messages per round w.h.p.
\end{lemma}

\begin{proof}
	Let $X_{v,r}$ be the number of messages node $v$ receives in round $r$. For bounding the probability of the events $X_{v,r} \leq a \cdot m$ for all $v \in V$ we are dealing with the balls into bins problem with $m\cdot n$ balls and $n$ bins. The claim for a single round follows from applying Lemma \ref{lem:balls_in _bins}.
	By the union bound given in Lemma~\ref{lem:unionbound} the event \smash{$\bigcap_{ r \leq p(n)} \big(\forall v \in V\!:  X_{v,r} \!\leq\! am \big)$} also takes place w.h.p.
\end{proof}


\section{k-Wise Independent Hash Functions}
\label{apx:Pseudorandom}

We will use families of $k$-wise independent hash functions, which are defined as follows.

\begin{definition}
	\label{def:hashfunctions}	
	For finite sets $A,B$, let $\mathcal H$ be a family consisting of hash functions $h : A \to B$. Then $\mathcal H$ is called $k$-wise independent if for a random function $h \in \mathcal H$ and for any $k$ distinct keys $a_1, \ldots, a_k \in A$ we have that $h(a_1), \ldots, h(a_k) \in B$ are independent and uniformly distributed random variables in $B$.
\end{definition}

From literature we know that such a family of hash functions exists in the following form (c.f., \cite{Vadhan2012}).

\begin{lemma}
	\label{lem:hashfunctions}
	For $A := \{0,1\}^a$ and $B := \{0,1\}^b$, there is a family of $k$-wise independent hash functions $\mathcal H := \{h : A \to B\}$ such that selecting a random function from $\mathcal H$ requires $k \!\cdot\! \max(a,b)$ random bits and computing $h(x)$ for any $x \in A$ can be done in $poly(a,b,k)$ time.
\end{lemma}

%

\begin{remark}
	We can use a random member $h \in \mathcal H$ of a family as described in Definition \ref{def:hashfunctions} to limit the number of messages any node receives in a given round. If all nodes send at most $\bigO(\log n)$ messages to targets that are determined using that hash function $h$ with a distinct key for each message, then any node receives at most $\bigO(\log n)$ messages per round w.h.p. The details are given in Lemma \ref{lem:receiveBound}.
\end{remark}

\section{Skeleton Graphs in Hybrid Networks}
\label{sec:skeleton-graphs}

Skeleton graphs have first been deployed by \cite{Ullman1991} and are a very useful concept in the context of the \hybrid model see \cite{Augustine2020a, Kuhn2020}. The basic idea of skeleton graphs is to sample a set of nodes of the local graph $G$ with some probability $\frac{1}{x}$ and then compute virtual edges among pairs of sampled nodes that are connected by a path of at most $h \in \tilO(x)$ hops, where the weight of that virtual edge corresponds to the length of that path. The formal definition is given as follows.

\begin{definition}
	\label{def:skeleton_graph}
	A skeleton graph $\calS = (V_\calS, E_\calS)$ of $G$, is obtained by sampling each node of $G$ to $V_\calS$ with prob.\ at least $\frac{1}{x}$. The edges of $\calS$ are $E_\calS \!=\! \{ \{u,v\} \!\mid\! u,v\!\in\!V_\calS, \text{hop}(u,v) \!\leq\! h\}$ with weights $d_h(u,v)$ for $\{u,v\} \in E_\calS$, for some appropriate $h\in \tilO(x)$.
\end{definition}

In \cite{Augustine2019} shows that $\calS$ gives a good approximation of the topology of the graph, in particular, with high probability, the distance between sampled nodes in the resulting skeleton graph equals the actual distance in the local graph.

\begin{lemma}[cf.\ \cite{Augustine2019}]
	\label{lem:skeleton-graph}
	A skeleton graph $\calS = (V_\calS, E_\calS)$ as given in Definition \ref{def:skeleton_graph} is connected and for any $u,v\in \calS$ it is $d_\calS(u,v) = d(u,v)$ w.h.p. for some appropriately chosen $h\in \tilO(x)$. Furthermore, the skeleton graph can be constructed in $h\in \tilO(x)$ rounds in the \LOCAL (and thus \HYBRID) model.
\end{lemma}

\section{The Node Communication Problem}
\label{sec:node_comm}

Our universal lower bounds is based on a reduction of an intermediate problem called the \textit{node communication problem}, which describes the complexity of communicating information between distinct node sets in the \hybrid model.

In the node communication problem we have two sets of nodes $A$ and $B$, where nodes in $A$ ``collectively know'' the state of some random variable $X$ and need to communicate it to $B$. 
The maximum amount of information that can be conveyed from $A$ to $B$ via global communication is fundamentally restricted by the nodes within $h$ hops that each node set can rely on for global communication.

\begin{definition}[Knowledge of Random Variables]
	\label{def:node_knowledge}
	Let $V$ be the set of nodes in a distributed network. Let $A \subseteq V$ and let $S_A$ be the state (including inputs) of all nodes in $A$ (we interpret $S_A$ as a random variable). Then the nodes in $A$ \textit{collectively know} the state of a random variable $X$ if $H(X|S_A) = 0$ (see Definition \ref{def:entropy}), i.e., there is no new information in $X$ provided that $S_A$ is already known.		
	Similarly, we say that $X$ is \textit{unknown} to $B \subseteq V$, if $H(X|S_B) = H(X)$, meaning that all information in $X$ is new even if $S_B$ is known. Equivalently, we can define this as $S_B$ and $X$ being independent. We can extend these definitions to communication parties Alice and Bob with states $S_{\text{Alice}}$ and $S_{\text{Bob}}$ where Alice knows $X$ if $H(X|S_{\text{Alice}}) = 0$ and $X$ is unknown to Bob if $H(X|S_{\text{Bob}}) =  H(X)$.
\end{definition}

We can now set up the node communication problem.

\begin{definition}[Node Communication Problem]
	\label{def:node_comm_problem}
	Let $G=(V,E)$ be some graph. Let $A,B \subset V$ be disjoint sets of nodes and $h := hop(A,B)$. Furthermore, let $X$ be a random variable whose state is collectively known by the nodes $A$ but unknown to any set of nodes disjoint from $A$. An algorithm $\mathcal A$ solves the \emph{node communication problem} if the nodes in $B$ collectively know the state of $X$ after $\mathcal A$ terminates. We say $\calA$ has success probability $p$ if $\calA$ solves the problem with probability at least $p$ for \emph{any} state $X$ can take.
\end{definition}

The Shannon entropy of a random variable $X$ can be thought of as the average information conveyed by a realization of $X$ and is defined as follows.

\begin{definition}[Entropy, c.f., \cite{Shannon1948}]
	\label{def:entropy}
	The Shannon entropy of a random variable $X\!:\! \Omega \!\to\! S$ is defined as $H(X) := - \!\sum_{x \in S} \mathbb{P}(X \!=\! x) \log \big(\mathbb{P}(X \!=\! x) \big)$. For two random variables $X,Y$ the \textit{joint entropy} $H(X,Y)$ is defined as the entropy of $(X,Y)$. The \textit{conditional entropy} is $H(X|Y) = H(X,Y) - H(Y)$. 
\end{definition}

\end{document}